\documentclass[11pt,intlim,righttag]{article}
\usepackage[centertags]{amsmath}
\usepackage{color}
\usepackage{amsthm,enumerate}
\usepackage[psamsfonts]{amsfonts,amssymb}
\usepackage{mathrsfs}
\usepackage{graphicx}
\usepackage[english]{babel}

\newcommand\beq{\begin{equation}}
\newcommand\eeq{\end{equation}}

\newtheorem{theorem}{Theorem}[section]

\newtheorem{proposition}{Proposition}[section]
\newtheorem{remark}{Remark}[section]

\title{Forward Backward SDEs Systems for Utility Maximization in Jump Diffusion Models}

\bigskip
\author{Marina Santacroce\thanks{Dipartimento di Matematica per le Scienze Economiche, Finanziarie ed Attuariali, Universit\`a Cattolica del Sacro Cuore, Via Necchi 9, 20123 Milano, Italy , \textit{E-mail:} \texttt{marina.santacroce@unicatt.it}} , Paola Siri and Barbara Trivellato\thanks{
\noindent  Dipartimento di Scienze Matematiche ``G.L. Lagrange'', Politecnico di Torino, Corso Duca degli Abruzzi 24, 10129 Torino, Italy. \textit{E-mails:}  \texttt{paola.siri@polito.it},
\texttt{barbara.trivellato@polito.it}}}
\date{}

\begin{document}

\maketitle

\begin{abstract}
	
\noindent We consider the classical problem of maximizing the expected utility of terminal net wealth with a final random liability in a simple jump-diffusion model. In the spirit of \cite{HHIRZ} and \cite{ST}, under suitable conditions the optimal strategy is expressed in implicit form in terms of a forward backward system of equations. Some explicit results are presented for the pure jump model and for exponential utilities.

	\bigskip
	

	
\end{abstract}

\section{Introduction}

	Portfolio's optimization and hedging problems are two classical problems which have been deeply investigated since the beginning of the seventies.
	Their mathematical formulation in continuous time was pioneered by Merton. In \cite{M} he provides the strategy which maximizes the expected utility of a small investor in closed form by exploiting the Markovianity of the model for standard utilities.
	An alternative approach for not necessarily Markovian models was suggested by \cite{B} and uses convex duality. This methodology has been developed in its full potential in  \cite{P},\cite{HP}, \cite{KLSX}, \cite{CH} and, in its modern form, in \cite{KS} and many others works thereafter.
It is well known that convex duality and martingale methods lead to establish the existence and uniqueness of an optimal strategy in general market models and for non classical utility functions, but they do not provide a constructive characterization.\\
	A constructive form for the optimal portfolio has been determined for classical utilities in quite general market models by using dynamic programming techniques (\cite{HIM}, \cite{MSc}), \cite{MS}). Nevertheless, there exist just very few results for non classical utilities.
	 A first work in this direction is given by \cite{MT03} (see also \cite{MT10}). In the non-Markovian framework, the problem is solved by dynamic programming and the optimal strategy is given in terms of the value function related to the problem and its derivatives. In turn, the value function is characterized as the solution of a backward stochastic partial differential equation under some assumptions imposed on the value function. Another subsequent work with a constructive characterization is \cite{HHIRZ}, where the hedging problem is studied for general utilities in a Brownian model without resorting to dynamic programming. In this article the optimal strategy is described by means of the utility function with its derivatives and the solution of a fully coupled forward backward system.
	In \cite{ST} the same approach has been generalized to a continuous semimartingale setting.
Further developments can be found in more recent works in Brownian settings (see \cite{NS} for a large investor problem with endogenous permanent market impacts and \cite{W} for a stochastic maximum principle algorithm for constrained utilities maximization).\\
    In this paper, we consider a jump diffusion model driven by a Brownian motion and an independent simple Poisson process and study a classical problem of maximization for a general utility using the techniques introduced in \cite{HHIRZ} and \cite{ST}.
	To our present knowledge, it is the first paper dealing with this approach in a model with jumps.\\
	We start by restating Proposition 2.1 in \cite{ST}, which gives the  derivative of the expectation of the terminal net wealth computed in an admissible strategy. The proof follows the same lines as in \cite{ST}, but we add here an extra condition which was missing. The derivative is proved to be null for the optimal strategy and this represents our necessary condition for optimality.\\
	In our main result we use this key condition to express the optimal strategy in terms of a fully coupled forward-backward stochastic differential system.
Besides, the vice versa of this result states that from a suitable solution of the system, the strategy which satisfies a certain optimality equation is admissible and optimal.
As a result, in the jump diffusion model the optimal strategy is found to be implicitely defined by an equation written in terms of the solution of the forward backward system. Such optimality equation can not be explicitely solved even for the exponential utility, where the system decouples (see e.g., \cite{Mo1}, \cite{Mo2}, \cite{Mo3}). However, the optimal strategy admits an explicit expression in the pure jump model. In this setting, for the pure investment problem we also provide sufficient conditions in order to find a solution to the forward backward system, which is not in general an easy task. Moreover, the related optimal strategy turns out to have a more pleasant expression.  
Finally, we prove that these sufficient conditions are satisfied for exponential utilities, where an easy solution to the forward backward system and the corresponding optimal strategy are explicitly written. 
	\\
	The paper is organized as follows.  In Section 2, we introduce the model with the assumptions and state the preliminary proposition, containing the necessary condition for the optimality. In Section 3, we consider the jump diffusion model and give the characterization of the optimal strategy in terms of a fully-coupled forward backward system. 
In Section 4, we deal with the pure jump model. The forward backward characterization of the optimal strategy is followed by the application to the pure investment problem.
	Finally, the results of the previous sections are specified to exponential utilities in Section 5.

\section{Model settings and preliminary results}
\noindent We consider the classical problem of maximizing the expected
utility of terminal net wealth, with a finite time horizon, when the asset price
process is modeled by a jump-diffusion. \\Specifically, the financial model consists of a bank account which pays no interests and of one risky asset whose (discounted) price $S$ is given by the stochastic exponential of a jump diffusion.

\noindent
 On a probability space $(\Omega,{\mathscr F},\mathbb{P})$ we consider two independent processes $W$ and $N$, both defined on $[0,T]$, where $T <+\infty$ is the fixed time horizon: $W$ is a standard (one-dimensional) Brownian motion and $N$ a homogeneous  Poisson process with intensity $\nu>0$.

\noindent
The probability space is equipped  with the natural filtration generated by $W$ and $N$ (completed by the
 $\mathbb{P}$-null sets of ${\mathscr F}$). For the sake of  simplicity, we denote the filtration by  ${\mathcal F}=({\mathcal F_t, t\in[0,T]})$ and we suppose $\mathcal F_T=\mathscr F$.

\noindent
We assume the  price process $S$  defined on the filtered probability space $(\Omega,{\mathscr F}, \mathcal F, \mathbb{P})$ with a dynamics given by

\begin{equation}\label{price}
	dS_t=S_{t^-}(\mu_tdt+\sigma_tdW_t+\eta_tdn_t),\quad S_0>0,
\end{equation}
where $n_t=N_t-\nu t$ defines the compensated Poisson martingale. We suppose  that the coefficients $\mu,\sigma,\eta$  are uniformly bounded predictable processes  and,  to ensure that $S$ is almost surely positive,  $\eta>-1$.
Under these assumptions on the coefficients of the model there is no arbitrage in the market.\\
 We investigate the model in two main different cases: in the first one we consider $\sigma^2>0$, whereas the second concerns the pure jump model, i.e. $\sigma^2 \equiv 0$.

\noindent
We denote by $\pi$ the dollar amount of risky asset in the portfolio and consider the wealth process $X^{\pi}$  which evolves according to the self-financing strategy $\pi$, i.e.
\begin{equation}\label{wealth}
X_t^{\pi}=X_0^{\pi}+\int_0^t \pi_s \frac{dS_s}{S_{s^-}}.
\end{equation}

\noindent
Moreover, let $H$ be a bounded ${\mathcal F}_T$-measurable random variable, representing a
liability due at time $T$. 

\noindent
Given a utility function $U$ defined on the real line, we consider the related problem of maximizing the expected utility of the terminal net wealth
\begin{equation}\label{problem}
to \; maximize\;\;\;\;\; \mathbb{E}[
U(X_T^{\pi}+H)]\;\;\;\;over \;
all\;\;\;\;\pi\in\Pi_x,
\end{equation}
where $x>0$ is a given initial capital and $\Pi_x$ is the set of admissible strategies defined as
\begin{equation}\label{Pi}
\Pi_x=\{\pi\in \mathbb H^2 \, \text{s.t.}\, X_0^{\pi}=x\}, \notag
\end{equation}
with
$$
\mathbb H^2 =\left\{ \vartheta\,\, \text{predictable process s.t.}\,\, \,\mathbb{E}\left(\int_0^T \vartheta^2_t dt\right)<+\infty \right\}.
$$
 Taking into account the boundedness of $\mu,\sigma$ and $\eta$, this choice ensures the square integrability of the wealth process \eqref{wealth}.


\noindent
Throughout the paper, the  utility function $U$ is assumed to be strictly increasing, strictly concave and
three times continuously differentiable. Moreover, the following additional conditions are used as needed:
\begin{description}
\item{(H1)} $\exists \, k>0$ s.t. the absolute risk aversion satisfies
$\text{ARA}(x)=-\frac{U''(x)}{U'(x)}\ge k$, $\forall \, x \in \mathbb R$;
\item{(H2)} $\mathbb{E}[(U^{\prime}(\xi)^2]<+\infty$, for a suitable random variable $\xi$;
\item{(H3)}  $\mathbb{E}[|U(\xi)|]<+\infty$, for a suitable random variable $\xi$.
\end{description}
From now on,  according to \eqref{wealth}  we consider the process $X^{0,h}$, where $X^{0,h}_0=0$ and $h$ is a predictable bounded process, i.e.
	\begin{equation}
		X_t^{0,h}=\int_0^t h_s \frac{dS_s}{S_{s^-}}.\notag
	\end{equation}
	Let us observe that, under our assumptions on the model and due to the boundedness of $h$, the pocess $X^{0,h}$ is square integrable.\\
	\noindent
	The following proposition represents the key starting point for the derivation of the main result and is a restatement of Proposition 2.1 in \cite{ST}. For this reason, we omit the proof. We just remark that we add here an extra condition  which was missing in \cite{ST}.

\begin{proposition}\label{Prop1}
 For $\pi^* \in \Pi_x$, let (H2) hold with  $\xi=X^{\pi^*}_T+H$. Moreover, suppose that for any predictable bounded process  $h$, there exists $\epsilon>0$ such that (H3) is satisfied by $\xi=X^{\pi^*+\epsilon h}_T+H$.

\noindent Then

\begin{equation}
\label{1}
\lim_{\epsilon \downarrow 0 }
\frac{\mathbb{E}[U(X^{\pi^*+\epsilon
h}_T+H)-U(X^{\pi^*}_T+H)]}{\epsilon}=\mathbb{E}\left[U^{\prime}(X^{\pi^*}_T+H)\, X^{0,h}_T\right].
\end{equation}
If, in addition, $\pi^*$ is optimal for the problem \eqref{problem}, then
\begin{equation}
\label{eg1}
\mathbb{E}\left[U^{\prime}(X^{\pi^*}_T+H)\, X^{0,h}_T\right]= 0.
\end{equation}

\end{proposition}
\begin{remark} Note that the extra integrability condition (H3) required in Proposition \ref{Prop1} can be weakened by requiring only the integrability of the negative part of $U(\xi)$.
\end{remark}

\section{Jump-diffusion model}

\noindent
In this section, we consider the dynamics  of the price process \eqref{price} with $\sigma^2>0$.

\noindent
The next two theorems are the main results of the paper. The first one  characterizes the optimal strategy  in terms of the solution of a forward-backward system of SDEs. The second theorem represents the vice versa and establishes the existence of an optimal strategy.

\begin{theorem}\label{Theo1}
Let $\pi^* \in \Pi_x$ be optimal for the problem \eqref{problem} and  suppose (H1) holds. Under the assumptions of Proposition \ref{Prop1}, there exists a smooth function $G$ such that the following forward-backward system admits a solution $(X,Y,Z,\Psi)$:
\begin{align}
\label{X1}X_t=&x+\int_0^tG(X_{s^-},Y_{s^-}, Z_s, \Psi_s,\Upsilon_s)(\mu_s ds+\sigma_sdW_s+\eta_sdn_s) \\[10pt]
\notag \label{Y1}Y_t=H+&\int_t^T\Big(\frac{U'(\Psi_s+G(X_{s^-},Y_{s^-},Z_s, \Psi_s,\Upsilon_s)\eta_s+X_{s^-}+Y_{s^-})-U'(X_{s^-}+Y_{s^-})}{U''(X_{s^-}+Y_{s^-})} -\Psi_s\Big)\nu ds\\    \notag +&\int_t^T\frac12 \frac{U'''(X_{s^-}+Y_{s^-})}{U''(X_{s^-}+Y_{s^-})}(Z_s+G(X_{s^-},Y_{s^-}, Z_s, \Psi_s,\Upsilon_s)\sigma_s)^2 ds, \\+&\int_t^T G(X_{s^-},Y_{s^-}, Z_s, \Psi_s,\Upsilon_s)\left(\mu_s -\eta_s\nu\right) ds -\int_t^T\!(Z_sdW_s+\Psi_sdn_s),
\end{align}
where  $\Upsilon_t=(\eta_t,\mu_t,\sigma_t)$.

\noindent Moreover,
\begin{equation}
\label{pi*}\pi^*=G(X_{-},Y_{-}, Z, \Psi,\Upsilon)
\end{equation}
 and the related optimal wealth process is equal to $X$.
\end{theorem}

\begin{proof}
  The arguments used in this proof are similar to those in \cite{HHIRZ} and \cite{ST}.  We start by defining
  \begin{equation}\label{alfa_def}
  \alpha_t=\mathbb{E}[U^{\prime}(X^{\pi^*}_T+H)|\mathcal{F}_t].\notag
  \end{equation}
Since $\alpha$ is a (square integrable) martingale, it satisfies the backward stochastic differential equation (BSDE)
  \begin{equation}\label{alfa}\alpha_t=U^{\prime}(X^{\pi^*}_T+H)-\int_t^T (\beta_sdW_s+\gamma_sdn_s),\notag
  	\end{equation}
  where $\beta$ and $\gamma$ are respectively the predictable integrand appearing in the martingale representation of $\alpha$, with respect to $W$ and $n$.
We consider the process $Y$, where
$$Y_t=(U')^{-1}(\alpha_t)-X_t^{\pi^*},$$
with final value $Y_T=H$. By Ito's formula we can write
\begin{align}\label{Y}
dY_t=&\frac{1}{U''(U'^{-1}(\alpha_{t^-}))}d\alpha_t-\frac12 \frac{U'''(U'^{-1}(\alpha_{t^-}))}{\big(U''(U'^{-1}(\alpha_{t^-}))\big)^3}d[\alpha]^c_t
+(U'^{-1}(\alpha_{t})-U'^{-1}(\alpha_{t^-})) \nonumber
\\&-\frac{1}{U''(U'^{-1}(\alpha_{t^-}))}\Delta\alpha_t-dX_t^{\pi^*}.
\end{align}
Observing that
$$
U'^{-1}(\alpha_{t})-U'^{-1}(\alpha_{t^-})=U'^{-1}(\gamma_t\Delta N_t+\alpha_{t^-})-U'^{-1}(\alpha_{t^-})=
\left(U'^{-1}(\gamma_t+\alpha_{t^-})-U'^{-1}(\alpha_{t^-})\right) \Delta N_t
$$
\eqref{Y} can be immediately rewritten as
\begin{align*}
	dY_t=&\frac{1}{U''(U'^{-1}(\alpha_{t^-}))}(\beta_t dW_t-\gamma_t \nu dt)-\frac12 \frac{U'''(U'^{-1}(\alpha_{t^-}))}{\big(U''(U'^{-1}(\alpha_{t^-}))\big)^3}\beta^2_tdt\\+&\left(U'^{-1}(\gamma_t+\alpha_{t^-})-U'^{-1}(\alpha_{t^-})\right) dN_t-\pi^*_t(\mu_tdt+\sigma_tdW_t+\eta_tdn_t).\\
\end{align*}
After rearranging the $dt$, $dW_t$ and $dn_t$ terms together and replacing $U'^{-1}(\alpha_t)=X_t^{\pi^*}+Y_t$, we denote the integrands in the martingale terms respectively by $Z_t$ and $\Psi_t$, i.e. we define the processes
 \begin{equation}\label{Zpsi}
Z=\frac{1}{U''(X_{-}^{\pi^*}+Y_{-})}\beta-\pi^*\sigma, \quad  \Psi=-\pi^*\eta+U'^{-1}\left(\gamma+U'(X_{-}^{\pi^*}+Y_{-})\right)-\left(X_{-}^{\pi^*}+Y_{-}\right).	
\end{equation}
Therefore, we can see that $Y$ solves the following BSDE
\begin{align*}
	dY_t=&Z_tdW_t+\Psi_tdn_t+\Big[	 U'^{-1}(\gamma_t+U'(X_{t^-}^{\pi^*}+Y_{t^-}))-(X_{t^-}^{\pi^*}+Y_{t^-})-\frac{1}{U''(X_{t^-}^{\pi^*}+Y_{t^-})}\gamma_t\Big]\nu dt\\& -\Big[\frac12 \frac{U'''(X_{t^-}^{\pi^*}+Y_{t^-})}{\big(U''(X_{t^-}^{\pi^*}+Y_{t^-})\big)^3}\beta^2_t+	\pi^*_t\mu_t\Big]dt,\qquad \qquad Y_T=H.
\end{align*}
If in the previous equation we replace $\beta$ and  $\Psi$ using the expressions in \eqref{Zpsi}, we obtain
\begin{align}\label{dYgamma}
\notag	dY_t=&-\Big[\frac12 \frac{U'''(X_{t^-}^{\pi^*}+Y_{t^-})}{U''(X_{t^-}^{\pi^*}+Y_{t^-})}(Z_t+\pi_t^*\sigma_t)^2-\left( \Psi_t+\pi^*_t\eta_t\right)\nu+\frac{1}{U''(X_{t^-}^{\pi^*}+Y_{t^-})}\gamma_t \nu +	\pi^*_t\mu_t\Big]dt\\ & +Z_tdW_t+\Psi_tdn_t, \quad Y_T=H.
\end{align}

\noindent In order to obtain the BSDE  \eqref{Y1} for $Y$ we apply \eqref{eg1} of Proposition \ref{Prop1}  to $\pi^*$.

\noindent Using the integration by parts formula
we have
\begin{align}\label{U'alfa}
  U^{\prime}(X^{\pi^*}_T+H)  X^{0,h}_T  = & \int_0^TU'(X_{t^-}^{\pi^*}+Y_{t^-}) h_t\big( \mu_tdt+\sigma_t dW_t+\eta_t dn_t\big)  \notag  \\
 &  +\int_0^T  X^{0,h}_{t^-} \, (\beta_tdW_t+\gamma_tdn_t)
  +\int_0^T \frac{h_t}{S_{t^-}}\left(\beta_td[S,W]_t+\gamma_td[S,n]_t\right) \notag \\
= & \int_0^T \alpha_{t^-} h_t\big(\sigma_t dW_t+\eta_t dn_t\big) +\int_0^T X^{0,h}_{t^-} \, (\beta_tdW_t+\gamma_tdn_t)
+\int_0^Th_t\eta_t\gamma_tdn_t \notag  \\
& +\int_0^T h_t\left( U'(X_{t^-}^{\pi^*}+Y_{t^-}) \mu_t +\beta_t\sigma_t+\gamma_t\eta_t \nu\right)dt.
\end{align}
We now check that the first three integrals in \eqref{U'alfa} are martingales.\\
For the first integral we have
\begin{align*}
\mathbb{E}\left(\underset{0\le t\le T}\sup\left|\int_0^t \alpha_{s^-} h_s\big(\sigma_s dW_s+\eta_s dn_s\big)\right|\right)\le& C\,\mathbb{E}\left(\left(\int_0^T \alpha_{t^-}^2 h_t^2\big(\sigma_t^2 dt+\eta_t^2 dN_t\big)\right)^{\frac12}\right)\\
\le& C\,\mathbb{E}\left(\underset{0\le t\le T}\sup|\alpha_{t^-}| \left(\int_0^T h_t^2\big(\sigma_t^2 dt+\eta_t^2 dN_t\big)\right)^{\frac12}\right)\\
\le& C\left(\mathbb{E}\left(\underset{0\le t\le T}\sup |\alpha_{t^-}|\right)^2 \right)^{\frac12} \left(\mathbb{E}\left(\int_0^T h_t^2 \big(\sigma_t^2 dt+\eta_t^2 dN_t\big)\right)\right)^{\frac12}\\
\le & C\left(\mathbb{E}\left(\alpha_{T-}^2 \right)\right)^{\frac12}  < \infty,
\end{align*}
\noindent where the constant $C$ can differ from line to line. The first is Burkholder-Davis-Gundy inequality for $p=1$, the third is  Cauchy-Schwarz inequality, while the last is a consequence of Doob inequality and the boundedness  of $h$, $\eta$ and $\sigma$. Then the conclusion follows by the integrability assumptions on $\alpha$.

\noindent Exploiting the same chain of inequalities, for the second integral we can write

\begin{equation*}
\mathbb{E}\left(\underset{0\le t\le T}\sup\left|\int_0^t
X^{0,h}_{s^-}\right.\right. \!\big(\beta_s dW_s \left.\left.+\gamma_s dn_s\big)\right|\right) \leq C\,
 \left(\mathbb{E}\left(X^{0,h}_{T-}\right)^2 \right)^{\frac12} \left(\mathbb{E}\left(\int_0^T \big(\beta_s^2 ds+\gamma_s^2 dN_s\big)\right)\right)^{\frac12} <\infty,
\end{equation*}
where the finiteness of the last expression is due to  the square integrability of $X^{0,h}$ and $\alpha$.

\noindent We are left to prove that $\int_0^\cdot h_s\eta_s\gamma_sdn_s$ is a martingale. A sufficient condition is to show $\mathbb{E}\left(\int_0^T |h_t\eta_t\gamma_t| \nu\, dt\right)<\infty$  for the predictable process $h\eta\gamma$. This is true since $h$ is bounded and by Cauchy-Schwarz
\begin{align*}
\mathbb{E}\left(\int_0^T |\eta_t\gamma_t| \, dt\right)\leq C &	\left(\mathbb{E}\left(\int_0^T \eta_t^2dt \right)\right)^\frac12 \left(\mathbb{E}\left(\int_0^T\gamma_t^2dt\right)\right)^\frac12 < \infty.
\end{align*}

\noindent Taking the expectation in \eqref{U'alfa} and recalling \eqref{eg1} we find
\begin{align}\label{exp}
 \notag \mathbb{E}&\left( U^{\prime}(X^{\pi^*}_T+H)\, X^{0,h}_T\right)=  \mathbb{E}\left(\int_0^Th_t\left( U'(X_{t^-}^{\pi^*}+Y_{t^-}) \mu_t
 +\beta_t\sigma_t+\gamma_t\eta_t \nu\right)dt\right)\\  = \mathbb{E}&\left(\int_0^T h_t\left( U'(X_{t^-}^{\pi^*}+Y_{t^-}) \mu_t
 +U''(X_{t^-}^{\pi^*}+Y_{t^-})(Z_t+\pi^*_t
 \sigma_t)\sigma_t+\gamma_t\eta_t \nu\right)dt\right)=0,
\end{align}
\noindent
where in the last equality we replaced $\beta=U''(X_{-}^{\pi^*}+Y_{-})(Z+\pi^*\sigma)$.\\
\noindent Choosing in \eqref{exp} the integrand  \,\,$h=1\!\!1_{\{U'(X_{-}^{\pi^*}+Y_{-}) \mu
	+U''(X_{-}^{\pi^*}+Y_{-})(Z+\pi^*\sigma)\sigma+\gamma\eta \nu>0\}}$ and then \\ $h=1\!\!1_{\{U'(X_{-}^{\pi^*}+Y_{-}) \mu
	+U''(X_{-}^{\pi^*}+Y_{-})(Z+\pi^*\sigma)\sigma+\gamma\eta \nu<0\}}$,
 we deduce
\begin{equation}\label{eqopt}
	U'(X_{-}^{\pi^*}+Y_{-}) \mu
	+U''(X_{-}^{\pi^*}+Y_{-})(Z+\pi^*\sigma)\sigma+\gamma\eta \nu=0 \qquad d\mathbb{P}\otimes dt-\text{a.e.}\,\, \text{on}\,\,[0,T].
	\end{equation}
%

\noindent From \eqref{Zpsi},
$$\gamma=U'(\Psi+\pi^*\eta+X_{-}^{\pi^*}+Y_{-})-U'(X_{-}^{\pi^*}+Y_{-})$$
which plugged into \eqref{eqopt} gives the equation
\begin{equation}\label{eqopt2a}	U'(X_{-}^{\pi^*}+Y_{-}) \mu
+U''(X_{-}^{\pi^*}+Y_{-})(Z+\pi^*\sigma)\sigma+\left(U'(\Psi+\pi^*\eta+X_{-}^{\pi^*}+Y_{-})-U'(X_{-}^{\pi^*}+Y_{-})\right)\eta\nu=0.
	\end{equation}

\noindent Note that \eqref{eqopt2a}  includes the case $\eta=0$ for which the optimal strategy can be explicitly written,
$$
\pi^*=\frac1\sigma\left(-\frac{U'(X_{-}^{\pi^*}+Y_-)}{U''(X_{-}^{\pi^*}+Y_-)}\frac\mu\sigma-Z\right),
$$
as in \cite{HHIRZ}.

\noindent
For any $\omega \in \Omega$  and  $t\in [0,T]$, \eqref{eqopt2a} can be seen as $F=0$, where the function $F:D\times{\mathbb R}\longrightarrow {\mathbb R}$, defined by
\begin{equation} \label{FDini}
F (w,\pi)=U'(x+y)\mu +U''(x+y)(z\sigma+\pi\sigma^2)+
\left(U'(\psi+\pi\eta+x+y)-U'(x+y)\right)\eta\nu
\end{equation}	

	$\forall w\in D=\{w=(x,y,z, \psi, \eta, \mu, \sigma)\in {\mathbb R}^7 : \sigma^2>0 \}$, 
	and $\forall \pi \in {\mathbb R}$,
	is continuously differentiable on its domain, because of the smoothness hypotheses on $U$.
	Moreover $$\frac{\partial}{\partial \pi} F (w,\pi)=
	U''(x+y)\sigma^2+U''(\psi+\pi\eta+x+y)\eta^2\nu < \frac 1 2 U''(x+y)\sigma^2<0
	 \ \ \forall (w,\pi)\in D\times{\mathbb R}.$$
	Since $D$ is open, a global implicit function theorem can be applied (see Appendix).  Consequently, for any $w\in D$ the equation $F (w,\pi)=0$ admits a unique solution $\pi\in {\mathbb R}$ and there exists a function $G$ continuously differentiable on $D$ and such that
	$\{(w,\pi)\in D\times {\mathbb R} : F (w,\pi)=0\}=
	\{(w,\pi)\in D\times {\mathbb R} : \pi=G(w)\}.$\\
Thus, the optimal strategy is
\begin{equation*}
	\pi^*= G(X_{-},Y_{-},Z,\Psi,\eta,\sigma,\mu),
\end{equation*}
which is expressed in terms of the solution $(Y,X,Z,\Psi)$ of the forward backward system  (\ref{X1}-\ref{Y1}), where the backward equation \eqref{Y1} is obtained replacing \eqref{pi*} in \eqref{dYgamma}. Besides, we note that \eqref{X1} represents the optimal wealth process.
\end{proof}

\medskip
\noindent
For the next result we require $\sigma^{-1}$ to be bounded in order to get the admissibility of the strategy $\pi^*$.

\begin{theorem}\label{converse}
Suppose a solution of system (\ref{X1}-\ref{Y1}) exists with  $Z\in \mathbb H^2$  and $G$ a smooth function such that $F(w,G(w))=0$, where $F$ is defined by \eqref{FDini}.  Assume (H1) and (H2), (H3) with $\xi= X_T+H$. Moreover, let $U^{\prime}(X+Y)$ be a positive martingale such that
the stochastic integral $\frac 1{U^{\prime}(X_{-}+Y_{-})} \cdot U^{\prime}(X+Y)$ is a square integrable martingale.
Then the optimal strategy $\pi^*$ exists in the class $\Pi_x$ and is given by \eqref{pi*}.

\end{theorem}
\begin{proof}

We first check that the strategy defined in \eqref{pi*} and satisfying \eqref{eqopt} is in $\Pi_x$.
 We consider $(X^{\pi^*},Y,Z,\Psi)$ satisfying the system (\ref{X1}-\ref{Y1}) and we write Ito's formula for $U'(X^{\pi^*}+Y)$, i.e.
\begin{equation*}\label{U'}
dU'(X^{\pi^*}_t+Y_t)=U'(X^{\pi^*}_{t^-}+Y_{t^-})\left(\frac{U''(X^{\pi^*}_{t^-}+Y_{t^-})}{U'(X^{\pi^*}_{t^-}+Y_{t^-})}(\pi^*_t \sigma_t+Z_t)dW_t+\frac{{\gamma_t}}{U'(X^{\pi^*}_{t^-}+Y_{t^-})} dn_t\right)
\end{equation*}
which, using  \eqref{eqopt}, can be rewritten as
\begin{equation} \label{ItoUprime} dU'(X^{\pi^*}_t+Y_t)=U'(X^{\pi^*}_{t^-}+Y_{t^-})\left(-\Big(\frac{\mu_t}{\sigma_t}+\frac{\eta_t\nu}{\sigma_t}\frac{{\gamma_t}}
{U'(X^{\pi^*}_{t^-}+Y_{t^-})}\Big)dW_t+\frac{\gamma_t}{U'(X^{\pi^*}_{t^-}+Y_{t^-})} dn_t\right).
\end{equation}
Since we have assumed that $\frac 1{U'(X^{\pi^*}_{-}+Y_{-})} \cdot U'(X^{\pi^*}+Y)$  is a square integrable martingale, by the martingale representation theorem and the hypotheses on boundedness of $\eta, \mu$ and $\sigma^{-1}$ we deduce that $\frac{{\gamma}}{U'(X^{\pi^*}_{-}+Y_{-})} \in \mathbb{H}^2$.
Using this fact, together with the hypotheses (H1) on ARA and $Z\in \mathbb{H}^2$, since \eqref{eqopt} can be rewritten as
\begin{equation*}
\pi^{*}=- \frac{1}{\sigma} \frac{U'(X_{-}^{\pi^*}+Y_{-})}{U''(X_{-}^{\pi^*}+Y_{-})} \left( \frac{\mu}{\sigma}  + \frac{\eta \nu }{\sigma} \frac{{\gamma}}{U'(X^{\pi^*}_{-}+Y_{-})} \right)-\frac{Z}{\sigma},
\end{equation*}
we get that $\pi^{*}\in \Pi_x$.

\noindent
Now we are are left to show that $\pi^{*}$ is optimal. From \eqref{ItoUprime}, we observe that the positive martingale $U'(X^{\pi^*}+Y)$ can be written as the Doleans exponential  of the martingale
\begin{equation}
\label{hatL}
M=-\Big(\frac{\mu}{\sigma}+\frac{\eta\nu}{\sigma}\frac{{\gamma}}{U'(X^{\pi^*}_{-}+Y_{-})}\Big)\cdot W+\frac{{\gamma}}{U'(X^{\pi^*}_{-}+Y_{-})}\cdot n,\notag
\end{equation}
where $\frac{{\gamma}}{U'(X^{\pi^*}_{-}+Y_{-})}>-1$.

\noindent Since $U$ is a concave function, for any $\pi\in \Pi_x$ we have that
\begin{equation}
  \label{U}
  U(X_T^\pi+H)-U(X_T^{\pi^*}+H)\le U'(X_T^{\pi^*}+H)\left(X_T^\pi -X_T^{\pi^*}\right).
\end{equation}

\noindent
Thus we define the measure $\mathbb{Q}$ equivalent to $\mathbb{P}$ by
\begin{equation}\label{Q}\frac{d\mathbb{Q}}{d\mathbb{P}}=\frac{U'(X_T^{\pi^*}+H)}{\mathbb{E}\big(U'(X_T^{\pi^*}+H)\big)},\notag
\end{equation}
 and taking the expectation in \eqref{U} we write
\begin{align}
\notag \mathbb{E}\left( \frac{U(X_T^\pi+H)-U(X_T^{\pi^*}+H)}{\mathbb{E}\big(U'(X_T^{\pi^*}+H)\big)}\right)&\le \mathbb{ E^Q}\left(X_T^\pi -X_T^{\pi^*}\right)\\\label{optimal}&=\mathbb{ E^Q}\left(\int_0^T\big(\pi_t -\pi_t^*\big)\frac{dS_t}{S_{t^-}}\right)=0,
\end{align}
where the last equality is a consequence of the Girsanov Theorem. In fact, $\frac{1}{S_{-}}\cdot S$ is a $\mathbb{Q}$- martingale since its predictable quadratic covariation with $M$ is
\begin{align*}
\left\langle \frac{1}{S_{-}}\cdot S, M \right \rangle &=\left \langle \sigma \cdot W+\eta \cdot n, -\Big(\frac{\mu}{\sigma}+\frac{\eta\nu}{\sigma}\frac{{\gamma}}{U'(X^{\pi^*}_{-}+Y_{-})}\Big)\cdot W+\frac{{\gamma}}{U'(X^{\pi^*}_{-}+Y_{-})}\cdot n \right \rangle \\
&= - \int_0^{\cdot}\mu_s ds.
\end{align*}
Moreover, since $\pi$ and $\pi^*$ are in $\Pi_x$, we get that
$$\int_0^T\big(\pi_t -\pi_t^*\big)\frac{dS_t}{S_{t^-}}=\int_0^T(\pi_t-\pi_t^*)(\mu_tdt+\sigma_tdW_t+\eta_tdn_t)$$
 is a $\mathbb{Q}$-local martingale. Finally,  by the following chain of inequalities we see that it is  a true $\mathbb{Q}$-martingale:
\begin{align*}\mathbb{ E^Q}\left(\underset{0\le t\le T}\sup\Big|\int_0^{t}\big(\pi_s -\pi_s^*\big)\frac{dS_s}{S_{s^-}}\Big|\right)\le C \, &  \mathbb{ E^Q}\left(\Big(\int_0^{T}\big(\pi_t -\pi_t^*\big)^2(\sigma^2_t+\eta^2_t\nu) dt\Big)^\frac12\right)\\=
C \,& \mathbb{E}\left(\frac{U'(X_T^{\pi^*}+H)}{\mathbb{E}\big(U'(X_T^{\pi^*}+H)\big)}\Big(\int_0^{T}\big(\pi_t -\pi_t^*\big)^2dt\Big)^\frac12\right)\\\le C & \frac{\left(\mathbb{E}\big((U'(X_T^{\pi^*}+H))^2\big)\right)^\frac12}{\mathbb{E}\big(U'(X_T^{\pi^*}+H)\big)}\left(\mathbb{E}\left(\int_0^{T}\big(\pi_t -\pi_t^*\big)^2dt\right)\right)^\frac12< +\infty.
\end{align*}

\noindent From \eqref{optimal}, we have
$$\mathbb{E}\left( U(X_T^\pi+H)-U(X_T^{\pi^*}+H)\right)\le 0, \qquad \text{for any}\, \pi\in \Pi_x$$
which means that $\pi^*$ is optimal. The proof is complete.
\end{proof}

\begin{remark}\label{REM2}
It is worth
noting that assumption (H2) in Theorem \ref{converse} which requires \,\,$\mathbb{E}[(U^{\prime}(X^{\pi^*}_T+H))^2]<+\infty$  is only used to show that,  for any $\pi\in \Pi_x$, $ \frac{\pi}{S_{-}}\cdot S $ is a $\mathbb{Q}$-martingale. Therefore, it can be replaced by any other condition which ensures this requirement.

\end{remark}

\section{Pure jump model}
 In this section we study the case where the asset price is a pure jump process and show that for this model $\pi^*$ can be written in an explicit form in terms of the forward backward SDE system solution.\\
We consider \eqref{price} with $\sigma^2 \equiv 0$, that is the price $S$  is modeled  by a pure jump process whose dynamics  is
\begin{equation}
	dS_t=S_{t^-}(\mu_tdt+\eta_tdn_t),\quad S_0>0.\notag
\end{equation}
We recall that the filtration  ${\mathcal F}=({\mathcal F_t, t\in[0,T]})$ is generated by a standard Brownian motion $W$ and a simple Poisson process $N$ with intensity $\nu>0$.   Moreover, the coefficients $\mu$ and $\eta$ are predictable and bounded with $\eta>-1$ and with the additional assumption $c_1\le\frac{\mu}{\eta}\le c_2<\nu$.

\noindent

\begin{theorem}\label{Theo3pj} 	
	Let $\pi^* \in \Pi_x$ be optimal for the problem \eqref{problem} and  suppose (H1) holds.
	Under the assumptions of Proposition \ref{Prop1}, there exists a  solution $(X,Y,Z,\Psi)$ of the following forward-backward system
	\begin{align}
	\label{Xpj}X_t=&x+\int_0^t\left(
U'^{-1}\left(U'(X_{s^-}+Y_{s^-})\left(1-\frac{\mu_s}{\eta_s \nu}\right) \right)-\left(X_{s^-}+Y_{s^-}+\Psi_s\right)\right)\left(\frac{\mu_s}{\eta_s} ds+dn_s\right) \\[10pt]
	\label{Ypj} Y_t=&H-\int_t^T\left[	\left(U'^{-1}\left(U'(X_{s^-}+Y_{s^-})\left(1-\frac{\mu_s}{\eta_s \nu}\right) \right)-\left(X_{s^-}+Y_{s^-}\right)\right)\right]\left(1-\frac{\mu_s}{\eta_s \nu}\right)\nu ds\\\notag  -\int_t^T& \left[\left( \Psi_s+\frac{U'(X_{s^-}+Y_{s^-})}{U''(X_{s^-}+Y_{s^-})}\right) \frac{\mu_s}{\eta_s}-\frac12 \frac{U'''(X_{s^-}+Y_{s^-})}{U''(X_{s^-}+Y_{s^-})}Z_s^2\right]ds-\int_t^T (Z_sdW_s+ \Psi_s dn_s).
		\end{align}
	Moreover, $\pi^*$ takes on the form
	\begin{equation}\label{pi*pj}
		\pi^*=\frac1{\eta}\left(
	U'^{-1}\left(U'(X_{-}^{\pi^*}+Y_{-})\left(1-\frac{\mu}{\eta \nu}\right) \right)-(\Psi+X_{-}^{\pi^*}+Y_{-})\right),
	\end{equation}
	and the optimal wealth process is equal to $X$.\\
	Vice versa, suppose a solution of the system (\ref{Xpj}-\ref{Ypj}) exists and assume (H1) and (H2), (H3) with $\xi= X_T+H$. Moreover, let $U^{\prime}(X+Y)$ be a positive martingale.
If the strategy \eqref{pi*pj} is in the class $\Pi_x$ then it is optimal.
\end{theorem}
\begin{proof}
Using the same arguments as those for the jump-diffusion case, we can write
\begin{align}\label{dYpj}
\notag	dY_t=&-\Big[\frac12 \frac{U'''(X_{t^-}^{\pi^*}+Y_{t^-})}{U''(X_{t^-}^{\pi^*}+Y_{t^-})}Z_t^2-\left( \Psi_t+\pi^*_t\eta_t\right)\nu+\frac{1}{U''(X_{t^-}^{\pi^*}+Y_{t^-})}\gamma_t \nu +	\pi^*_t\mu_t\Big]dt\\ & +Z_tdW_t+\Psi_tdn_t, \quad Y_T=H.
\end{align}
Condition \eqref{eqopt} now becomes
\begin{equation}\label{eqoptpj}
	U'(X_{-}^{\pi^*}+Y_{-}) \mu +\gamma\eta \nu=0 \qquad d\mathbb{P}\otimes dt-\text{a.e.}\,\, \text{on}\,\,[0,T],
	\end{equation}
where
\begin{equation}\label{gammadef}
\gamma=U'(\Psi+\pi^*\eta+X_{-}^{\pi^*}+Y_{-})-U'(X_{-}^{\pi^*}+Y_{-}).
\end{equation}
From \eqref{eqoptpj} and \eqref{gammadef}, we deduce  \eqref{pi*pj}. Finally, plugging $\gamma$ obtained from $ \eqref{eqoptpj}$ and \eqref{pi*pj} into \eqref{dYpj}, we deduce \eqref{Ypj}.
\\
For the converse, let us observe that the positive martingale $U'(X^{\pi^*}+Y)$ can be written as the Doleans exponential  of the martingale
\begin{equation}
\label{hatLpj}
M=\frac{U^{''}(X^{\pi^*}_{-}+Y_{-})}{U'(X^{\pi^*}_{-}+Y_{-})} Z \cdot W-\frac{{\mu}}{\eta \nu}\cdot n,
\end{equation}
whose predictable quadratic covariation with $\frac{1}{S_{-}}\cdot S$ is
$
\langle \frac{1}{S_{-}}\cdot S, M  \rangle  = - \int_0^{\cdot}\mu_s ds.
$
Thus, by Girsanov Theorem, $\frac{1}{S_{-}}\cdot S$ is a martingale with respect to the $\mathbb{P}$-equivalent measure $\mathbb{Q}$ defined by
$$\frac{d\mathbb{Q}}{d\mathbb{P}}=\frac{U'(X_T^{\pi^*}+H)}{\mathbb{E}\big(U'(X_T^{\pi^*}+H)\big)}.$$
The conclusion for the optimality of $\pi^*$  follows as in Theorem \ref{converse}.
\end{proof}

\noindent
As previously said, it is not easy to find a solution of the coupled system. In the next proposition, we consider the pure investment  problem, i.e. $H=0$, and find a sufficient condition for obtaining an explicit solution to the system (\ref{Xpj} - \ref{Ypj}) with $Z=0$. At the end of the next section, we show an example where this condition is satisfied.
	

\begin{proposition} \label{Proppj} Let $H=0$ and suppose (H1) holds.
Consider the system
	\begin{align}
		\label{XpjH0}X_t=&x+\int_0^t  \frac{1}{ARA(X_{s^-})} \frac{a_s-{\mu_s}/{\eta_s}} {\mu_s-\eta_s \nu} \left(\mu_s ds+ \eta_s dn_s \right), \\[10pt]
		\label{YpjH0} Y_t=& U'^{-1} \left(U'(X_{t})e^{A_t}\right) -X_t,
	\end{align}
where $a\in \mathbb{H}^2$  and $A_t=-\int _t^T a_s \, ds$, with $A_0$ deterministic.  If
 $U'(X)e^{A}$ is a positive martingale,
then $(X,Y)$ gives a  solution to the forward backward system (\ref{Xpj} - \ref{Ypj}) with $Z=0$.
Moreover, under (H3) with $\xi= X_T$ the strategy
	\begin{equation}
	\pi^*=\frac{1}{ARA(X_{-})} \frac{a-{\mu}/{\eta}} {\mu-\eta \nu}
	\end{equation}
	is in $\Pi_x$ and it is optimal. \\
	\end{proposition}

\begin{proof}
	
	Applying  Ito's formula to $U'(X_{t})e^{A_t}$, and taking into account the definition of $A_t$, we get
	\begin{align}\label{Uprimeexp1}
		d\left(U'(X_{t})e^{A_t}\right)&= U'(X_{t^{-}})e^{A_t} a_t dt + U^{''}(X_{t^{-}})e^{A_t}\frac{1}{ARA(X_{t^-})} \left( a_t-\frac{\mu_t}{\eta_t} \right)dt \notag \\
		& \quad + U'(X_{t})e^{A_t} - U'(X_{t^{-}})e^{A_t}\notag\\
		&= U'(X_{t^{-}})e^{A_t}\left(\left(\frac{\mu_t}{\eta_t}+\frac{\nu\gamma_t}{e^{A_t} U'(X_{t^{-}})}\right)dt+\frac{\gamma_t}{e^{A_t} U'(X_{t^{-}})}dn_t\right)
	\end{align}
where $\gamma_t \Delta N_t=e^{A_t} \left(U'(X_{t})-U'(X_{t^{-}})\right)$.\\
Since we assumed that $U'(X+Y)=U'(X)e^A$ is a  martingale, we find $\gamma = - \frac{\mu}{\eta \nu} e^{A} U'(X_{{-}})$  and therefore
\begin{equation}\label{Uprimeexp2}
d\left(U'(X_{t})e^{A_t}\right)=-U'(X_{t^{-}})e^{A_t}\frac{\mu_t}{\eta_t \nu}dn_t.
\end{equation}

\noindent
From \eqref{YpjH0} we get $Y_T=0$ and, by Ito's formula, using \eqref{Uprimeexp1} we see that
	\begin{align}\label{YpjH01}
		dY_t=& \frac{U'(X_{t^-})e^{A_t}}{ U^{''}\left(U'^{-1} \left(U'(X_{t^-})e^{A_t}\right)\right)}  \frac{\mu_t }{\eta_t }  dt +
\left[ U'^{-1} \left(U'(X_t)e^{A_t}\right)-U'^{-1} \left(U'(X_{t^-})e^{A_t}\right)\right] \notag
\\ &   -\frac{1}{ARA(X_{t^-})} \frac{a_t-{\mu_t}/{\eta_t}} {\mu_t-\eta_t \nu} \left(\mu_t dt+ \eta_t dn_t \right)\, \notag \\
=& \frac{U'(X_{t^-})e^{A_t}}{ U^{''}\left(U'^{-1} \left(U'(X_{t^-})e^{A_t}\right)\right)}  \frac{\mu_t }{\eta_t }  dt +
\left[ U'^{-1} \left(U'(X_{t^-})e^{A_t}\left(1-\frac{\mu_t}{\eta_t \nu}\Delta N_t\right)\right)-U'^{-1} \left(U'(X_{t^-})e^{A_t}\right)\right] \notag
\\ &   +\frac{1}{ARA(X_{t^-})} \, \frac{\mu_t -\eta_t a_t}{\mu_t-\eta_t\nu}\left(\frac{\mu_t}{\eta_t}\, dt + dn_t\right)\, \notag \\
=& \frac{U'(X_{t^-})e^{A_t}}{ U^{''}\left(U'^{-1} \left(U'(X_{t^-})e^{A_t}\right)\right)}  \frac{\mu_t }{\eta_t }  dt +
\left[ U'^{-1} \left(U'(X_{t^-})e^{A_t}\left(1-\frac{\mu_t}{\eta_t \nu}\right)\right)-U'^{-1} \left(U'(X_{t^-})e^{A_t}\right)\right] \Delta N_t\notag
\\ &   +\frac{1}{ARA(X_{t^-})} \, \frac{\mu_t -\eta_t a_t}{\mu_t-\eta_t\nu}\left(\frac{\mu_t}{\eta_t}\, dt + dn_t\right)\, \notag \\
=& \left( \frac{U'(X_{t^-})e^{A_t}}{ U^{''}\left(U'^{-1} \left(U'(X_{t^-})e^{A_t}\right)\right)} +\frac{1}{ARA(X_{t^-})} \, \frac{\mu_t -\eta_t a_t}{\mu_t-\eta_t\nu}  \right) \frac{\mu_t }{\eta_t }  dt \notag  \\
 &  + \left[ U'^{-1} \left(U'(X_{t^-})e^{A_t}\left(1-\frac{\mu_t}{\eta_t \nu}\right)\right)-U'^{-1} \left(U'(X_{t^-})e^{A_t}\right)\right] \left(1-\frac{\mu_t}{\eta_t \nu}\right) \nu dt\notag \\
 &  + \left[ U'^{-1} \left(U'(X_{t^-})e^{A_t}\left(1-\frac{\mu_t}{\eta_t \nu}\right)\right)-U'^{-1} \left(U'(X_{t^-})e^{A_t}\right)\right] \frac{\mu_t}{\eta_t} dt\notag \\
&   + \left[ \frac{1}{ARA(X_{t^-})} \, \frac{\mu_t -\eta_t a_t}{\mu_t-\eta_t\nu}+U'^{-1} \left(U'(X_{t^-})e^{A_t}\left(1-\frac{\mu_t}{\eta_t \nu}\right)\right)-U'^{-1} \left(U'(X_{t^-})e^{A_t}\right)\right] dn_t,\notag
	\end{align}
which corresponds to
\begin{align*}
	dY_t=&\left[	\left(U'^{-1}\left(U'(X_{t^-}+Y_{t^-})\left(1-\frac{\mu_t}{\eta_t \nu}\right) \right)-\left(X_{t^-}+Y_{t^-}\right)\right)\right]\left(1-\frac{\mu_t}{\eta_t \nu}\right)\nu dt\\\notag  +
 & \left( \Psi_t+\frac{U'(X_{t^-}+Y_{t^-})}{U''(X_{t^-}+Y_{t^-})}\right) \frac{\mu_t}{\eta_t} dt +  \Psi_t dn_t
		\end{align*}
with
$$\Psi= \frac{1}{ARA(X_{-})} \, \frac{\mu -\eta a}{\mu-\eta \nu}+U'^{-1} \left(U'(X_{-})e^{A}\left(1-\frac{\mu}{\eta \nu}\right)\right)-U'^{-1} \left(U'(X_{-})e^{A}\right).$$
Therefore, $Y$ satisfies   \eqref{Ypj} with $Z=0$.

\noindent
The admissibility of $\pi^*$ follows from (H1), the assumptions on $a_t$ and on the model coefficients. In order to check the optimality of the strategy, by the converse part of Theorem \ref{Theo3pj}, we are left to prove that  $U'(X)e^{A}$ is a positive martingale and $\mathbb{E}\left[\left(U'(X_T)\right)^2\right]<+\infty$.  But this is true since, by \eqref{Uprimeexp2},  $U'(X)e^{A}$ is the Doleans exponential  of the martingale
	$M=-\frac{{\mu}}{\eta \nu}\cdot n$
	whose predictable quadratic variation   $\langle M\rangle= \int_0^{\cdot} \frac{{\mu_s}^2}{\eta_s^2 \nu}\, ds $
	is  bounded. In fact, using Novikov condition,  $\mathcal E (2M)$ is a uniformly integrable martingale and, thus,
 $\left[\mathcal E _T(M)\right]^2=\mathcal E_T (2M) e^{\langle M\rangle_T }$ has finite 
 expectation.

\end{proof}

\section{Exponential utility}

\noindent
  In this last section we tailor the results obtained previously to the case of the exponential utility, which has been extensively studied in the literature (e.g., \cite{DGRSSS}, \cite{F}, \cite{MZ}, \cite{SST}). It can be easily seen that with this particular choice of utility function, the forward-backward system decouples and the problem reduces to the study of a backward stochastic differential equation. \\
For results for continuous price models, see, e.g., \cite{REK}, \cite{HIM} in Brownian setting, and  \cite{MSc}, \cite{Mo1} and \cite{MS} for more general continuous semimartingale models. For models allowing jumps, we quote among others \cite{Mo2}, \cite{Mo3}, \cite{AM}, and \cite{MeSt} for the pure jump model.\\
We consider the exponential utility
$U(x)=-e^{-\delta x}$
\noindent with risk aversion parameter $\delta \in (0,+\infty)$ and a bounded random liability $H$. In this case the evolution for
$$Y=-\frac{1}{\delta}\log\frac{\alpha}{\delta}-X^{\pi^*}$$
does not depend on the wealth process $X^{\pi^*}$ and is given by the following backward equation
$$
dY_t=\left((e^{-\delta (\Psi_t+\pi_t^*\eta_t)}-1)\frac{\nu}{\delta}+\frac{1}{2}\delta (Z_t+\pi_t^* \sigma_t)^2+(\Psi_t+\pi_t^*\eta_t)\nu
-\pi_t^*\mu_t\right)dt +Z_tdW_t+\Psi_tdn_t,
$$
with final condition $ Y_T=H$.

\noindent
Moreover, equation \eqref{eqopt2a} in the proof of Theorem \ref{Theo1} for the optimal strategy $\pi^*$ can be rewritten as
$$
\mu_t
-\delta(Z_t+\pi^*_t\sigma_t)\sigma_t+\left(e^{-\delta(\Psi_t+\pi^*_t\eta_t)}- 1 \right) \eta_t \nu=0
\qquad d\mathbb{P}\otimes dt-\text{a.e.}\,\, \text{on}\,\,[0,T],$$
from which we can deduce  that also the optimal strategy is independent on the wealth.\\
\noindent
Let us point out that, employing a notation similar to \cite{Mo1},  the driver of the BSDE can be rewritten in the form
$$
f_t(z,\psi)=\sup_{\pi\in \mathbb{R}}\left\{  -\frac \delta 2 \left(\pi \sigma_t - (-z +\frac{\mu_t}{\delta \sigma_t})\right)^2 - [-\psi -\pi \eta_t]_{\delta}\right\} -\frac{\mu_t}{\sigma_t} z + \frac{\mu_t^2}{2 \delta \sigma_t^2},
$$
where, for $\psi\in \mathbb{R}$,
$
[\psi]_\delta=\frac{\nu}{\delta}\left(e^{\delta \psi}-1 -\delta \psi\right).
$
Then, using the results in \cite{Mo1}, we deduce that there exists a solution to the BSDE with $Y$ bounded and $Z, \Psi\in \mathbb{H}^2$.

\medskip
\noindent
In the pure jump case, the backward evolution of $Y$ reduces to
\begin{equation*}
	dY_t=\left(\frac 12 \delta Z_t^2 -\frac \nu \delta \left(1-\frac{\mu_t}{\eta_t\nu}\right)\ln \left(1-\frac{\mu_t}{\eta_t\nu}\right) + \frac{\mu_t}{\eta_t}\left( \Psi_t  -\frac 1 \delta \right) \right) dt +Z_tdW_t+\Psi_tdn_t
\end{equation*}
and the optimal strategy $\pi^*$ can be written in an explicit form as stated in the next proposition.

\begin{proposition}\label{prop_exp}
A solution of the backward equation 
		
\begin{equation}\label{Yexppj}
	Y_t=H-\int_t^T\left(\frac 12 \delta Z_s^2 -\frac \nu \delta \left(1-\frac{\mu_s}{\eta_s\nu}\right)\ln \left(1-\frac{\mu_s}{\eta_s\nu}\right) + \frac{\mu_s}{\eta_s}\left( \Psi_s  -\frac 1 \delta \right) \right) ds -\int_t^T(Z_sdW_s+\Psi_sdn_s)
\end{equation}
exists with $Y$ bounded and $Z, \Psi\in \mathbb{H}^2$	
and the strategy $\pi^*$
\begin{equation}\label{pi*exppj}
	\pi^*=-\frac 1 \eta \left(\frac 1 \delta \ln \left(1-\frac{\mu_t}{\eta_t\nu}\right) + \Psi\right)
\end{equation}
is in $\Pi_x$.
In addition, if  $e^{-\delta (X^{\pi^*}+Y)}$ is a positive martingale and either $Z$ or $\Psi$ is bounded, then \eqref{pi*exppj} is optimal.
\end{proposition}

\begin{proof}
	
The existence of a solution to \eqref{Yexppj} with $Y$ bounded and  $Z, \Psi \in \mathbb{H}^2$ is proved in Theorem 1 of Antonelli and Mancini \cite{AM}. This theorem requires two main assumptions on the driver $f$ of \eqref{Yexppj}, which we can rewrite as
\begin{description}
	\item[\sc{A1)}] $f$ is measurable and predictable and satisfies $d\mathbb{P} \otimes dt$ - a.e.
	\begin{equation}\label{A1}
		-\lambda_t-\frac \delta 2 z^2 -[-\psi]_\delta \leq f_t(z,\psi) \leq \lambda_t+\frac \delta 2 z^2 +[\psi]_\delta
	\end{equation}
	for a strictly positive predictable process $\lambda$ such that\, $\text{essup}_\omega \int_0^T \lambda_t \, dt <+ \infty$;
	\item[\sc{A2)}] $f$  verifies
	\begin{equation}\label{A2}
		f_t(z,\psi) - f_t(z,\psi^{\prime}) \leq \zeta_t^{z, \psi, \psi^{\prime}} (\psi -\psi^{\prime}) \nu
	\end{equation}
	where the process $\zeta^{z, \psi, \psi^{\prime}}$ is such that $D_1 \leq \zeta_t^{z, \psi, \psi^{\prime}} \leq D_2 $, with $-1< D_1 \leq 0$ and $D_2\geq 0$.
\end{description}

\noindent In order to check these assumptions we rewrite the driver $f$ in the form
\begin{equation}\label{jumpdriver}
f_t(z,\psi)=\sup_{\pi\in \mathbb{R}}\left\{ \pi \mu_t - [-\psi -\pi \eta_t]_{\delta}\right\} -\frac12 \delta z^2.
\end{equation}

\noindent Since $\pi=0 \in \mathbb{R}$, from \eqref{jumpdriver} we deduce the validity of the left inequality in \eqref{A1}. To prove the right inequality, we consider the function $F(\psi)=f_t(z,\psi) -\frac \delta 2 z^2 -[\psi]_\delta$. Using the explicit form of the driver in \eqref{Yexppj}, we can rewrite $F(\psi)$ as
$$
F(\psi)= \frac \nu \delta \left(\left(\ln \left(1-\frac{\mu}{\eta \nu}\right)  +1 + \delta \psi\right) -e^{\delta \psi}\right), \, \text{  }  
$$
whose maximum is attained at $\psi^*=\frac 1 \delta \ln \left(1-\frac{\mu}{\eta \nu}\right)$ and holds $$F(\psi^*)=2\, \frac \nu \delta \left(1-\frac{\mu}{\eta \nu}\right) \ln \left(1-\frac{\mu}{\eta \nu}\right).$$ We deduce that in \eqref{A1} we can take $\lambda= 2\, \frac \nu \delta \left(1-\frac{\mu}{\eta \nu}\right)\left| \ln \left(1-\frac{\mu}{\eta \nu}\right)\right|$, which is a bounded predictable and strictly positive process, thus the assumption A1) is verified. \\
Using again the explicit form of the driver,  \eqref{A2} can be rewritten in the form
$$
( \psi - \psi^{\prime}) \left( \zeta_t^{z, \psi, \psi^{\prime}} + \frac{\mu}{\eta \nu} \right)\geq 0,
$$ and A2) is easily verified choosing  $\zeta^{z, \psi, \psi^{\prime}}=- \frac{\mu}{\eta \nu} $, $D_1=-1$ and a suitable constant $D_2$, which can be found since $\frac{\mu}{\eta }$ is bounded. \\
\noindent
From the existence of the solution with $Y$ bounded and  $Z, \Psi \in \mathbb{H}^2$ and from \eqref{pi*exppj}, thanks to the standing assumptions on the model, we deduce  that $\pi^*\in \Pi_x$. 
\\ Finally, from the vice versa of Theorem \ref{Theo3pj},  the optimality of $\pi^*$  is guaranteed if we assume that either $\Psi$ or $Z$ is bounded. In fact, (H1) and (H3) trivially hold true,  whereas (H2)  reduces to check that
$\mathbb{E} \left[e^{-2 \delta (X_T^{\pi^*}+H)}\right]$ is finite. \\
First we suppose $\Psi$ bounded, so is $\pi^*$. Taking into account also $H$ boundedness, we get
$$\mathbb{E}\left[ e^{-2 \delta (X_T^{\pi^*}+H)}\right]\leq C_2\mathbb{E}\left[e^{C_1 N_T}\right] <+\infty,
$$
where $C_1, C_2$  are suitable constants whose specific values are irrelevant. \\
Assume now $Z$  bounded. Since the positive martingale $e^{-\delta (X^{\pi^*}+Y)}$ is the Doleans exponential of the martingale $M$ defined in \eqref{hatLpj},  we get
 $\left[\mathcal E_T (M)\right]^2=\mathcal E_T (2M) e^{\langle M\rangle_T}$. Moreover, if $Z$ is bounded then 
 $$\langle M\rangle_T= \int_0^{T} \delta^2 Z_s^2 \, ds + \int_0^{T} \frac{{\mu_s}^2}{\eta_s^2 \nu}\, ds $$
 is bounded.
 Using Novikov condition,  $\mathcal E (2M)$ is a uniformly integrable martingale and, therefore,
$$\mathbb{E}\left[ e^{-2 \delta (X_T^{\pi^*}+H)}\right]= \mathbb{E}\left[ \left[\mathcal E_T (M)\right]^2\right]<+\infty.
$$
\end{proof}

\begin{remark}{ \rm We notice that, in the case $Z$  bounded, the martingality assumption on $e^{-\delta (X^{\pi^*}+Y)}$ can be omitted since it is automatically satisfied.}\end{remark}

\begin{remark}{\rm In \cite{MeSt}, the exponential utility maximization problem is studied within a BSDE framework and in a L\'evy-driven  pure jump asset model. Existence  of optimal strategies are proved using BMO arguments and assumptions on the solutions of the BSDEs involved which imply the boundedness of the strategies. Although our market model is simpler, Proposition \ref{prop_exp}  represents an attempt to establish conditions for the existence of optimal strategies possibly not bounded.
	}
\end{remark}
\bigskip
\noindent

\noindent We conclude the section with an example for  the pure investment problem, where the hypotheses of Proposition \ref{Proppj} are automatically satified.
\begin{proposition}
If $H=0$ and $\frac{\mu}{\eta}$ is deterministic then the strategy 
\begin{equation*}
	\pi^*= - \frac{1 }{\delta \eta}\ln{\left(1-\frac{\mu}{\eta \nu}\right)}
\end{equation*}
is in $\Pi_x$ and it is optimal.
\end{proposition}	
\begin{proof}
We choose in Proposition \ref{Proppj},
	\begin{equation}\label{a}
		a=\left(\frac{\mu}{\eta \nu}+\left(1-\frac{\mu}{\eta \nu}\right)\ln{\left(1-\frac{\mu}{\eta \nu}\right)}\right) \nu,
	\end{equation}
which is a positive bounded process by the assumptions on the model.
	Then, the  processes $X$ and $Y$  become 
	 \begin{align*}
		X_t=&x+\int_0^t  \frac{1}{ARA(X_{s^-})} \frac{a_s-{\mu_s}/{\eta_s}} {\mu_s-\eta_s \nu} \left(\mu_s ds+ \eta_s dn_s \right) \\
		=&x+\int_0^t  \left(- \frac{1 }{\delta \eta_s}\right)\ln{\left(1-\frac{\mu_s}{\eta_s \nu}\right)}  \left(\mu_s ds+ \eta_s dn_s \right), \\[10pt]
		Y_t=& U'^{-1} \left(U'(X_{t})e^{A_t}\right) -X_t= -\frac 1 \delta A_t= \frac 1 \delta \int_t^T a_s ds.
	\end{align*}

\noindent On the other hand, we can observe that \eqref{a} represents the unique choice for $a$ in Proposition \ref{Proppj} which makes $U'(X)e^{A}$ a martingale. In fact, $U'(X)e^{A}$ is a martingale if and only if
\begin{align*}
	e^{A}\left( U'(X)-U'(X_{-}  \right)= - \frac{\mu}{\eta \nu} e^{A} U'(X_{{-}}) \Delta N,
\end{align*}
i.e.
\begin{align*}
	\displaystyle \delta e^{-\delta X_{-}+A}\left( e^{-\delta   \left( \frac{1 }{\delta}\frac{a-{\mu}/{\eta}} {\mu-\eta \nu} \eta \right)} -1 \right)= - \frac{\mu}{\eta \nu} \delta e^{-\delta X_{-}+A},
\end{align*}
from which we deduce $\frac{\mu-\eta a} {\mu-\eta \nu}=\ln{\left(1- \frac{\mu}{\eta \nu} \right)}$ and thus \eqref{a}. The positivity of $U'(X)e^{A}$ immediately follows from the conditions on the model.  

\noindent
Since all the assumptions of Proposition \ref{Proppj} are satisfied, the processes $X$ and $Y$
give a  solution to the forward backward system (\ref{Xpj} - \ref{Ypj}) with $Z=0$ and
\begin{align*}
	\Psi&= \frac{1}{ARA(X_{-})} \, \frac{\mu -\eta a}{\mu-\eta \nu}+U'^{-1} \left(U'(X_{-})e^{A}\left(1-\frac{\mu}{\eta \nu}\right)\right)-U'^{-1} \left(U'(X_{-})e^{A}\right) \\
	&= \frac 1 \delta \ln{\left(1-\frac{\mu}{\eta \nu}\right)} -\frac 1 \delta   \ln{\left(e^{-\delta X_{-}+A} \left(1-\frac{\mu}{\eta \nu}\right)  \right) }+\frac 1 \delta   \ln{\left(e^{-\delta X_{-}+A}  \right) } =0.
\end{align*}
Finally, the  strategy is
\begin{equation*}
	\pi^*= \frac{1}{ARA(X_{-})} \frac{a-{\mu}/{\eta}} {\mu-\eta \nu}= - \frac{1 }{\delta \eta}\ln{\left(1-\frac{\mu}{\eta \nu}\right)}.
\end{equation*}
\end{proof}

\begin{remark}{\rm We notice that in the proposition above the hypothesis on $\frac{\mu}{\eta}$  can be weakened by assuming deterministic $\int_0^T a_s ds$, with the choice of $a$ in \eqref{a}. }
\end{remark}

\noindent

\section{Appendix}
This result is a generalization of Lemma 1 in \cite{GW}.

\begin{theorem}
	Let $A$ be an open subset of ${\mathbb R^n}$.
	
	\noindent
	Let us consider a function $F\in C^1(A\times {\mathbb R};{\mathbb R})$ and suppose there exists a function $g:A\rightarrow {\mathbb R}$ such that one of the following conditions holds:
	\begin{enumerate}
	\item
	$\displaystyle\frac{\partial}{\partial y} F({\mathbf x},y)> g({\mathbf x})>0$,  $\forall ({\mathbf x},y)\in A\times {\mathbb R}$,
	\item
    $\displaystyle\frac{\partial}{\partial y} F({\mathbf x},y) < g({\mathbf x})<0 \ \forall ({\mathbf x},y)\in A\times {\mathbb R}$.
    \end{enumerate}
    Then there exists a function $G\in C^1(A; {\mathbb R})$ such that $F({\mathbf x}, G({\mathbf x}))=0$, $\forall {\mathbf x}\in A$.
	\end{theorem}

\begin{proof}
	We start by proving that for any $\mathbf x \in A$ there exists  $y(\mathbf x)\in \mathbb R$ such that $F({\mathbf x}, y({\mathbf x}))=0$.\\
	In order to do this, let us suppose that condition 1 holds. We can proceed in the same way if condition 2 is assumed.\\
	Let us fix  $\mathbf x \in A$. If $F({\mathbf x}, 0)=0$, there is nothing to check. Otherwise, by Lagrange mean value theorem,  we have
	
$$
F\left({\mathbf x}, \frac{|F({\mathbf x}, 0)|}{g({\mathbf x})}\right)=F({\mathbf x}, 0)+\frac{\partial}{\partial y}F({\mathbf x},y)\Big|_{y=y_1({\mathbf x})}  \frac{|F({\mathbf x}, 0)|}{g({\mathbf x})}> F({\mathbf x}, 0)+|F({\mathbf x}, 0)|\geq 0
$$
and
$$
F\left({\mathbf x}, -\frac{|F({\mathbf x}, 0)|}{g({\mathbf x})}\right)=F({\mathbf x}, 0)-\frac{\partial}{\partial y}F({\mathbf x},y)\Big|_{y=y_2({\mathbf x})} \frac{|F({\mathbf x}, 0)|}{g({\mathbf x})}<F({\mathbf x}, 0)+|F({\mathbf x}, 0)|\leq 0,
$$
where 	$ 0<y_1({\mathbf x})<\frac{|F({\mathbf x}, 0)|}{g({\mathbf x})}$ and
	$-\frac{|F({\mathbf x}, 0)|}{g({\mathbf x})}<y_2({\mathbf x})<0$.
Then, by the intermediate value theorem, there exists $-\frac{|F({\mathbf x}, 0)|}{g({\mathbf x})}<y({\mathbf x})< \frac{|F({\mathbf x}, 0)|}{g({\mathbf x})}$ such that $F({\mathbf x}, y({\mathbf x}))=0$.\\
Let us observe that, since $F({\mathbf x}, \cdot)$ is strictly monotone, $y({\mathbf x})$ is univocally determined.\\
We are left with the task of proving that the function $G$ defined by $G({\mathbf x})=y({\mathbf x})$ is continuously differentiable on $A$.
Let us fix ${\mathbf x_0}\in A$. Since $F({\mathbf x_0}, G({\mathbf x_0}))=0$, by the implicit function theorem we can find a neighborhood $U_{\mathbf x_0}$ of ${\mathbf x_0}$ and a function $\varphi_{\mathbf x_0}\in C^1(U_{\mathbf x_0}; \mathbb R)$ such that $F({\mathbf x}, \varphi_{\mathbf x_0}({\mathbf x}))=0$, for any ${\mathbf x}\in U_{\mathbf x_0}$.
Considering that also F$({\mathbf x}, G({\mathbf x}))=0$ and
 $F$ is strictly monotone, we deduce that $G= \varphi_{\mathbf x_0}$, therefore continuously differentiable,  on $U_{\mathbf x_0}$.
Due to the arbitrariness of ${\mathbf x_0}$ we can then conclude that $G\in C^1(A; {\mathbb R})$ and get the thesis.
	\end{proof}

\end{document}